\documentclass[12pt,a4paper]{article}
\textheight      = \paperheight
\advance\textheight by -2truein
\textwidth       = \paperwidth
\advance\textwidth  by -2truein
\usepackage[T1]{fontenc}

\usepackage{indentfirst,eurosym}
\usepackage{amsmath,amsthm,amsfonts,amscd,amssymb,eucal,a4wide}

\usepackage{ucs}
\usepackage[utf8x]{inputenc}

\def\RR{{\mathbb R}}
\def\CC{{\mathbb C}}
\def\NN{{\mathbb N}}
\def\ZZ{{\mathbb Z}}
\def\a{{\mathcal A}}

\def\h{{\mathcal H}}
\def\d{{\mathcal D}}
\def\diff{{{\rm Diff}^+ (S^1)}}

\newtheorem{theorem}{Theorem}[section]

\newtheorem{corollary}[theorem]{Corollary}
\newtheorem{proposition}[theorem]{Proposition}
\newtheorem{lemma}[theorem]{Lemma}

\parindent=\baselineskip
\def\parchar#1{\par\noindent
  \hbox to \parindent{#1\hfill}\ignorespaces
}

\begin{document}
\title{Local equivalence of representations of $\diff$ corresponding to different highest weights}
\author{Mih\'aly Weiner\footnote{Supported in part by the 
ERC advanced grant 669240 QUEST ``Quantum Algebraic Structures and Models'' and by OTKA grant no. 104206.}}
\date{}
\maketitle
\begin{abstract}
Let $c,h$ and $c,\tilde{h}$ be two admissible pairs of central charge and highest weight for $\diff$. It is shown here that the
positive energy irreducible projective unitary representations 
$U_{c,h}$ and $U_{c,\tilde{h}}$ of the group $\diff$
are {\it locally} equivalent. This means that for any $I\Subset S^1$ open proper interval, there exists a unitary operator $W_I$ such that $W_I U_{c,h}(\gamma)W_I^* = U_{c,\tilde{h}}(\gamma)$
for all $\gamma \in \diff$ which act identically on $I^c\equiv S^1\setminus I$ (i.e.\! which can ``displace'' or ``move'' points only in $I$).
This result extends and completes earlier ones that dealt with  only certain regions of the ``$c,h$-plane'', and closes the gap in the full classification of superselection sectors of Virasoro nets.
\end{abstract}

\section{Introduction}

The highest weight projective unitary representations of the group of orientation preserving diffeomorphisms $\diff$ of the unit circle $S^1=\{z\in \CC|\, |z|=1\}$ play a fundamental role in conformal quantum field theory. We postpone the detailed description of the representation $U_{c,h}$
associated to an admissible pair of the central charge $c>0$ and highest weight $h\geq 0$ (and how it is obtained from the unitary representation $L^{c,h}$ of the Virasoro algebra through the use of the stress-energy field $T_{c,h}$) to
the preliminaries, but note here that they are all irreducible and pairwise inequivalent: that is, if
$(c,h)$ and $(\tilde{c},\tilde{h})$ are both admissible pairs and
$W$ is a unitary such that $WU_{c,h}(\gamma)W^* = U_{\tilde{c},\tilde{h}}(\gamma)$
for all $\gamma\in\diff$, then $(c,h)= (\tilde{c},\tilde{h})$ and $W$ is a multiple
of the identity. However, some of these representations might be {\it locally} equivalent. This means, that even with $(c,h)\neq (\tilde{c},\tilde{h})$ it can happen that for any open proper interval of the circle $I\Subset S^1$, the restrictions to the subgroup formed by the diffeomorphisms localized in $I$ are unitarily equivalent; i.e.\! that for any $I\Subset S^1$ there exists a unitary $W_I$ such that
$W_I U_{c,h}(\gamma)W_I^* = U_{\tilde{c},\tilde{h}}(\gamma)$ for all
\begin{equation}
\gamma \in G_I = \{\gamma\in \diff | \, \gamma|_{S^1\setminus I} = {\rm id}_{S^1\setminus I}\}.
\end{equation}
Local equivalence can be also formulated at the level of self-adjoint generators:
as it will be explained in the preliminaries, the unitary $W_I$ establishes a local equivalence (relative to $I\Subset S^1$) between $U_{c,h}$ and $U_{\tilde{c},\tilde{h}}$ if and only if
\begin{equation}
W_I T_{c,h}(f)W_I^* = T_{\tilde{c},\tilde{h}}(f)
\end{equation}
for all $f\in C^\infty(S^1,\RR)$ with support in $I$.

The question of local equivalence comes up naturally when studying superselection sectors of conformal field theory in the setting \cite{FrG} of Haag-Kastler nets. When $(c,0)$ is admissible, the collection of von Neumann algebras 
\begin{equation}
\a_c(I)= \{U_{c,0}(\gamma)|\, \gamma\in G_I\}''\;\;\;\;\;\; (I\Subset S^1)
\end{equation}
together with the representation $U_{c,0}$ form a conformal chiral net: the {\it Virasoro net} Vir$_c$ at central charge $c$. It is a highly important model since every conformal chiral net of von Neumann algebras contains a Virasoro net as an irreducible subsystem; a fact which for example enabled complete classification \cite{KL} of conformal chiral nets with central charge $c<1$ and a partial one \cite{Carpi2,Xu} at central charge $c=1$. 

From the point of view of the Virasoro subnet, the full conformal net with central charge $c$ as well as its superselection sectors are all locally normal representations of Vir$_c$.
This is why it is so crucial to understand and classify the superselection sectors (i.e.\! the locally normal irreducible representations) of Vir$_c$. 

A locally normal irreducible representation $\pi$ of any conformal net $(\a,U)$ --- and in particular, of $(\a_c,U_{c,0})$ --- is automatically diffeomorphism covariant with positive energy \cite{weiner}. This means that we have a strongly continuous projective unitary positive energy representation $U^\pi$ of $\diff$ such that 
\begin{equation}
\pi_I(U(\gamma)) = U^\pi(\gamma)
\end{equation}
for every $\gamma\in G_I$. In case we deal with a Virasoro net, then $U^\pi$ must be irreducible and hence --- up to unitary equivalence --- it must coincide with one of the highest weight representations $U_{\tilde{c},\tilde{h}}$; see e.g.\! \cite[Theorem A.2]{Carpi2} and the references there given for the classification of strongly continuous projective unitary positive energy representations of $\diff$. As $\pi_I$ is always unitarily implementable --- see e.g.\! \cite[Lemma 4.4]{FrG} --- we finally arrive to the conclusion: sectors of Vir$_c$ are in one-to-one correspondence of the highest weight projective unitary representations of $\diff$ that are {\it locally equivalent} to $U_{c,0}$ in the sense we introduced it here, c.f.\! \cite[Proposition 2.1]{Carpi2}.

As is explained in the preliminaries, the central charge $c$ is ``locally detectable''. That is, if $U_{c,h}$ and $U_{\tilde{c},\tilde{h}}$ are locally equivalent, then $c=\tilde{c}$.
So we can discuss the problem for each possible value of the central charge in a separate manner.

The case $c<1$ has been ``completely cleared'': because of the coset construction of Goddard, Kent and Olive \cite{GKO}, we know that for each value of the highest weight $h$ for which the pair $(c,h)$ is admissible, $U_{c,h}$ indeed defines a superselection sector of Vir$_c$ (that is, it is locally equivalent to $U_{c,0}$); see the more detailed explanation at \cite[Section 2.4]{Carpi2}. Actually, for the $c<1$ case, not only that we have a complete list of sectors, but even their fusion rules and statistical dimensions are well-understood, see more in \cite{KL}.

Using stress-energy constructions in the vacuum representation space of the U(1) current algebra, Buchholz and Schulz-Mirbach proved \cite{BS-M} local equivalence of $U_{c,h}$ and
$U_{c,0}$ for all values of $c\geq 1$ and $h\geq \frac{c-1}{24}$. Some results concerning 
fusion rules and statistical dimensions of these sectors can be found in \cite{Rehren,Carpi1,Carpi2,Xu}; however, our knowledge is not complete. 
As noted by Buchholz, using tensorial products, one can show that if the admissible pairs
$(c,h)$ and $(\tilde{c},\tilde{h})$ are ``good'' in the sense that $U_{c,h}$ is locally 
equivalent to $U_{c,0}$ and likewise, $U_{\tilde{c},\tilde{h}}$ is locally 
equivalent to $U_{\tilde{c},0}$, then it follows that also $U_{c+\tilde{c},h+\tilde{h}}$
is ``good'' in that it is locally equivalent to $U_{c+\tilde{c},0}$; see the
details at \cite[Section 2.4]{Carpi2}. In this way the region where local equivalence 
of $U_{c,h}$ and $U_{c,0}$ can be shown enlarges; in particular all values of $c\geq 2$
and $h\geq 0$ will fall in. However, for example when $1<c<\frac{1}{2}+\frac{7}{10}$, 
this method will surely not give anything since the two smallest possible values of the central charge are $\frac{1}{2}$ and $\frac{7}{10}$. Thus, after many years the problem was first noted, some regions of the $c,h$-plane could still not be covered till now.

In this paper it is shown that if $c>1$, then $U_{c,h}$ is locally equivalent 
to $U_{c,h_c}$ for any $h<h_c=\frac{c-1}{24}$. Together with the listed earlier results
this completely settles the question of local equivalence and shows that 
for any two admissible pairs $(c,h)$ and $(\tilde{c},\tilde{h})$, the representations
$U_{c,h}$ and $U_{\tilde{c},\tilde{h}}$ are locally equivalent if and only if $c=\tilde{c}$.
In particular, each highest weight representation $U_{c,h}$ gives a locally normal irreducible representation of Vir$_c$ (and there are no other ones).

The proof relies on two main ingredients. First, just as the method of Buchholz and Schulz-Mirbach, it uses realizations of the Virasoro
algebra in the U(1) current (or as it is also called: the Heisenberg) algebra. Second,
that under certain conditions, the dependence of expectation values of the form 
\begin{equation}
\langle\Psi_{c,h}, \, 
e^{iT_{c,h}(f_1)} \ldots 
e^{iT_{c,h}(f_n)}\Psi_{c,h}
\rangle
\end{equation}
on $h$ (where $f_1,\ldots f_n\in C^\infty(S^1,\RR)$ are considered as fixed functions and $\Psi_{c,h}$ is the normalized highest weight vector) can be shown to be complex analytic. It is the method 
of analytic continuations that will ultimately allow us to access the region not covered by previous arguments.

\section{Preliminaries}

A unitary ``highest weight'' representation of Virasoro algebra with central charge $c>0$ and highest weight $h\geq 0$ consists of a complex scalar product space $V_{c,h}$ and a collection of linear operators $L^{c,h}_n$ $(n\in\ZZ)$ acting on $V_{c,h}$ such that
we have the Virasoro algebra commutation relations
\begin{equation}
{[} L_n,L_m{]} = (n-m) L_{n+m} +\frac{c}{12}(n^3-n)\delta_{n,-m}I\;\;\;\;\; (n,m\in\ZZ)
\end{equation}
the unitarity condition
\begin{equation} 
\langle u, L_n v\rangle = \langle L_{-n}u,v\rangle \;\;\;\;\; (u,v\in V_{c,h}, n\in \ZZ),
\end{equation}
and an up-to-phase 
unique normalized vector (which we shall refer to as the normalized highest weight vector) $\Psi_{c,h}\in V_{c,h},\; \|\Psi_{c,h}\|=1$ such that
\begin{equation}
L_0^{c,h}\Psi_{c,h} = 0\;\;\;\;\;\;\;\;
L^{c,h}_n\Psi_{c,h} = 0\;\;\; \textrm{for all}\;\;n>0
\end{equation}
and $V_{c,h}$ is the smallest subspace containing $\Psi_{c,h}$ and invariant for all operators $L_n^{c,h}$ $(n\in \ZZ)$. These representations are completely determined by the listed properties, are irreducible, and for different values of the central charge and highest weight are pairwise inequivalent in the following sense. If both $L^{c_1,h_1}_{n}$
and $\tilde{L}^{c_2,h_2}_n$ ($n\in\ZZ$) form a unitary highest weight representation of the
Virasoro algebra with central charges $c_1$ and $c_2$, highest weights $h_1$ and $h_2$ and
normalized highest weight vectors $\Psi_{c_1,h_1}$ and $\tilde{\Psi}_{c_2,h_2}$, respectively, then there exists an invertible linear map $W$ such that 
$WL^{c_1,h_1}_{n}W^{-1} = \tilde{L}^{c_2,h_2}_n$ for all $n\in \ZZ$ if and only if 
$c_1=c_2$ and $h_1=h_2$ and in this case $W$ can be uniquely normalized so that 
$W\Psi_{c_1,h_1}= \tilde{\Psi}_{c_2,h_2}$. Moreover, with this normalization $W$ is actually a scalar product preserving linear isomorphism. 

When such a unitary highest weight representation exists, we will say that $(c,h)$ is an {\it admissible pair}; this happens if and only if  
either $c\geq 1$ and $h\geq 0$ or there exists an $m\in\NN, \; m\geq 3$ and a $p,q\in \{1,\ldots m+1\}, \;q<p$
such that
\begin{equation}
c= 1-\frac{6}{m(m+1)}\;\;\;\;\;\textrm{and}\;\;\;\;\;
h=\frac{((m+1)p-mq)^2-1}{4 m(m+1)}.
\end{equation}
For more details, references and background on the representation theory of the Virasoro algebra, we refer to the book \cite{kac}. Here we are only interested by how such a representation ``integrates''  into a projective unitary representation of $\diff$. 

Let $(c,h)$ be an admissible pair and $\h_{h,c}= \overline{V}_{c,h}$ be the Hilbert space obtained by the completion of
$V_{c,h}$. It can be shown that for every smooth function $f:S^1\to \RR$ is a smooth function
with Fourier components
\begin{equation}
\hat{f}_n = \int_{-\pi}^{pi}f(e^{i\theta})e^{-in\theta}\frac{d\theta}{2\pi},\;\;\;\;
(n\in\ZZ)
\end{equation}
the sum
\begin{equation}
T_{c,f}^0(f) = \sum_{n\in\ZZ}\hat{f}_n L^{c,h}_n
\end{equation}
is absolute convergent on every vector of the dense subspace $V_{c,h}\subset \h_{c,h}$.
The obtained operator is closable, its closure $T_{c,h}(f)= \overline{T_{c,h}^0(f)}$ is 
self-adjoint. We shall refer to $T_{c,h}$ as the stress-energy field. It turns out that
$V_{c,h}$ is included in the domain of every product of the form $T_{c,h}(f_1) T_{c,h}(f_2)\ldots T_{c,h}(f_n)$ and with 
\begin{equation}
{[}f,g]:= f \partial_\theta g - f \partial_\theta g\;\;\;\;
\textrm{and} \;\;\; (f,g):= 
\int_{-\pi}^{\pi}\left(\frac{d}{d\theta} f(e^{i\theta}) + \left(\frac{d}{d\theta}\right)^3 f(e^{i\theta}) \right)g(e^{i\theta})\frac{d\theta}{2\pi}
\end{equation}
where $\partial_\theta f$ is the function $e^{i\theta}\mapsto \frac{d}{d\theta}f(e^{i\theta})$, one has that the algebraic relations 
\begin{eqnarray}
\nonumber
T_{c,h}(f)v + tT_{c,h}(g)v = T_{c,h}(f+ t g)v, \\ 
{[}T_{c,h}(f),T_{c,h}(g)]v = i T_{c,h}([f,g])v + \frac{ic}{12}(f,g)v
\end{eqnarray}
is satisfied on every vector $v$ of the dense subspace $V_{c,h}$. Actually, by considering scalar products and taking account of the fact that $V_{c,h}$ is a core for $T_{c,h}(q)$, from here it is an elementary exercise to show that 
\begin{eqnarray}
\nonumber
\overline{T_{c,h}(f) \, + \, t\,  T_{c,h}(g)} =  \, T_{c,h}(f+ t g), \\ 
\label{alg_rel}
\overline{{[}T_{c,h}(f),T_{c,h}(g)]} =   i \, T_{c,h}([f,g]) \, + \, \frac{ic}{12}(f,g) I.
\end{eqnarray}
Finally, one has \cite{GoWa} that there exists a unique projective unitary representation 
$U_{c,h}$ such that
\begin{equation}
U_{c,h}(\gamma) T_{c,h}(f) U_{c,h}(\gamma)^* = T_{c,h}(\gamma_* f) + r(c,f,\gamma) I
\end{equation}
for all $\gamma\in\diff$ and $f\in C^\infty(S^1,\RR)$. Here $\gamma_* f =
(\tilde{\gamma}' f)\circ \gamma^{-1}$ where $\tilde{\gamma}$ is a $2\pi$-periodic diffeomorphism of $\RR$ such that $e^{i\tilde{\gamma}(\theta)}=\gamma(e^{i\theta})$
for all $\theta \in \RR$, and $r(c,f,\gamma)$ is a certain
real constant (depending on $c,f$ and $\gamma$) whose exact value will be irrelevant 
for our discussion. The representation $U_{c,h}$ is strongly continuous and irreducible,
and considering an $f\in C^\infty(S^1,\RR)$ as the real vector field on $S^1$ symbolically written as $f(e^{i\theta})\frac{d}{d\theta}$, with some abuse 
of notations (as on one side a unitary, whereas on the other we shall put a projective unitary operator) we further have that
\begin{equation}
e^{iT_{c,h}(f)} = U_{c,h}({\rm Exp}(f)).
\end{equation}
Since $\diff$ is a simple group (see the survey \cite{Milnor} of Milnor and the references there given), every $\gamma\in \diff$ can be written as a finite product of exponentials 
(as such exponentials generate a normal subgroup) and so the above property actually uniquely determines the representation $U_{c,h}$. 

In the converse direction, it can also be shown \cite{Loke} that a strongly continuous projective unitary irreducible representation $U$ of $\diff$ which satisfies the ``positive energy'' requirement, is unitarily equivalent to one of these constructed highest weight representations. Here the positivity requirement means that there exists a positive
operator $A\geq 0$ such that $U(R_\alpha) = e^{i\alpha A}$ (where $R_\alpha$ is the rotation defined by the formula $R_\alpha(z) = e^{i\alpha}z$) for every $\alpha \in \RR$.
For more details and references, the reader should consult the appendix of \cite{Carpi2}.
We shall finish this section with some easy confirmations regarding local equivalence, c.f.\! \cite[Proposition 2.1]{Carpi2}. 
\begin{lemma}
Let $(c,h)$ and $(\tilde{c},\tilde{h})$ be two admissible pairs, $I\Subset S^1$ an open proper interval and $W_I$ a unitary operator. Then the two conditions
\begin{itemize}\label{lemmacc}
\item[i)] $W_I T_{c,h}(f) W_I^* = T_{\tilde{c},\tilde{h}}(f)$ for all $f\in C^\infty(S^1,\RR)$
with support in $I$,
\item[ii)]
$W_I U_{c,h}(\gamma) W_I^* = U_{\tilde{c},\tilde{h}}(\gamma)$ for all $\gamma\in G_I$.
\end{itemize}
are equivalent, and any of them implies that $c=\tilde{c}$.
\end{lemma}
\begin{proof}
The first condition clearly implies the second one because exponentials of vector fields with support in $I$ generate a dense subgroup in $G_I$, see \cite[Section V.2]{Loke}. Conversely, assume the second property. Then whenever $f\in C^\infty(S^1,\RR)$ has a 
support in $I$ and $t$ is a real number, $W_I e^{itT_{c,h}(f)} W_I^*= e^{itW_I T_{c,h}(f)W_I^*} $ must be a multiple of $e^{itT_{\tilde{c},\tilde{h}}(f)}$ and hence
 $W_I T_{c,h}(f)W_I^*$ and $T_{\tilde{c},\tilde{h}}(f)$ can only differ in an additive constant. Suppose now that $f_1,f_2,g_1,g_2\in C^\infty(S^1,\RR)$ have their support in $I$.
 Then using that additive constants do not matter inside a commutator, we find that 
 \begin{eqnarray}
 \nonumber
 W_I \left([T_{c,h}(f_1),T_{c,h}(g_1)] - [T_{c,h}(f_2),T_{c,h}(g_2)]\right) W_I^*
 \\
 = [T_{\tilde{c},\tilde{h}}(f_1),T_{\tilde{c},\tilde{h}}(g_1)] - [T_{\tilde{c},\tilde{h}}(f_2),T_{\tilde{c},\tilde{h}}(g_2)].
 \end{eqnarray}
 On the other hand, it is easy to see that one can choose $f_1,f_2,g_1,g_2$ is such a manner that $[f_1,g_1]=[f_2,g_2]$ but $(f_1,g_1)\neq (f_2,g_2)$. Then by the relations discussed at eq.\! (\ref{alg_rel}), the closure of the left hand side is equal to  $\frac{ic}{12}((f_1,g_1)-(f_2,g_2))I$, whereas that of the right hand side is equal to $\frac{ic}{12}((f_1,g_1)-(f_2,g_2))I$, implying that $c=\tilde{c}$. Once we know that $c=\tilde{c}$, we can use similar arguments (relying on commutators) to show that $W_I T_{c,h}([g_1,g_2])W_I^*$ is precisely equal (no additive constant) to $T_{\tilde{c},\tilde{h}}([g_1,g_2])$. However, elementary analysis shows that every 
 $f\in C^\infty(S^1,\RR)$ with support in $I$ can be written as the sum of at most $2$ commutators; i.e.\! that with suitable choice of the local functions $f_1,f_2,g_1,g_2$
 we have $f=[f_1,f_2]+[g_1,g_2]$.
\end{proof}

\section{Realizations of the Virasoro algebra using currents}

A unitary representation of the $U(1)$ current (or as it also called: the Heisenberg algebra) consist of a complex scalar product space $V$ and a collection of linear operators $J_n$ $(n\in\ZZ)$ acting on $V$ and satisfying the commutation relation
\begin{equation}
[J_n,J_m] = n \delta_{n,-m} I\;\;\;\;\; (n,m\in \ZZ)
\end{equation}
and the unitarity condition
\begin{equation}
\langle u, J_n v\rangle = \langle J_{-n}u,v\rangle \;\;\;\;\; (u,v\in V, n\in \ZZ).
\end{equation}
In what follows we shall suppose that we deal with the {\it vacuum representation} of $U(1)$ current algebra;
that is, we have an (up-to-phase unique) element of unit length $\Omega \in V$ (called the vacuum vector) such that 
\begin{equation}
J_n\Omega = 0 \;\;\;\textrm{for every}\;\;\; n\geq 0
\end{equation}
and $V$ is the minimal subspace containing $\Omega$ and invariant to all operators $J_n$ $(n\in\ZZ)$.
It then follows that the seemingly infinite sum appearing in 
\begin{equation}
L_n = \frac{1}{2}:\!J^2\!:_n \, \equiv\,  
\frac{1}{2}
\left(\sum_{m=-\infty}^{-1}J_m J_{n-m} + 
\sum_{m=0}^{\infty}J_{n-m}J_m\right)
\end{equation}
actually results in only finitely many non-zero terms whenever it is applied to a vector of $V$, and defines a unitary representation of the Virasoro algebra with unit central charge. Moreover, 
the hermitian $L_0$ is diagonalizable with nonnegative integer eigenvalues:
\begin{equation}
V= \oplus_{k=0}^\infty V_k\;\;\; \textrm{where}\;\;\;
V_k = {\rm Ker}(L_0- kI).
\end{equation}
Further, $J$ is covariant with respect to this Virasoro algebra representation:
\begin{equation}
[L_n,J_m] = -m\, J_{n+m}\;\;\;\; (n,m\in\ZZ).
\end{equation}
In a similar manner to how it was done in the preliminaries, one can {\it smear}
$J$ and $T$ with smooth test functions and for an $f\in C^\infty(S^1,\RR)$ introduce
\begin{equation}
J(f) = \overline{\sum_{n\in \ZZ}\hat{f}_nJ_n},\;\;\;\;\; 
T(f) = \overline{\sum_{n\in \ZZ}\hat{f}_nJ_n}
\end{equation}
which will again turn out to be self-adjoint operators. Moreover, one finds that every vector of $V$ is actually analytic for $J(f)$, the {\it Weyl-operator} $e^{iJ(g)}$
leaves invariant the dense subspace 
\begin{equation}
\cap_{n\in \NN}\d(\overline{L}^n_0)
\end{equation}
of {\it smooth vectors} and thus one can use convergent series 
to show that with $\partial_\theta g$ defined as the function $z=e^{i\theta} \mapsto \frac{d}{d\theta}g(e^{i\theta})$, we have the
transformation rules
\begin{equation}\label{trule1}
e^{iJ(g)} J(f) e^{-iJ(g)} = J(f) + \int_{-\pi}^{\pi}(\partial_\theta g)(e^{i\theta})f(e^{i\theta})\frac{d\theta}{2\pi}\;  I 
\end{equation}
and 
\begin{equation}\label{trule2}
e^{iJ(g)} T(f) e^{-iJ(g)} = T(f) + J((\partial_\theta g)f)+ 
\int_{-\pi}^{\pi}\frac{(\partial_\theta g)^2(e^{i\theta})}{2}f(e^{i\theta})\frac{d\theta}{2\pi}  \; I,
\end{equation}
see the details for example at \cite[Section 4.2]{CLTW}.

The following construction is well-known and has been used others, see e.g.\! \cite{kac,FH,Carpi2} (though note also that at \cite[Section 3.4]{kac}, the formula is given not on the vacuum space of the U(1) current). However, in part because of differences in conventions and notations, in part because of self-containment here we briefly recall the main idea. 
\begin{proposition}
For any pair of values $\alpha,\beta\in \CC$ the operators $\tilde{L}^{\alpha,\beta}_0=L_0 + \frac{1}{2}(\alpha^2+\beta^2) I$ and
$$
\tilde{L}^{\alpha,\beta}_n= L_n + \alpha J_n + i\beta n J_n
\;\;\;(n\in\ZZ,\, n\neq 0)
$$
form a representation of the Virasoro algebra with central charge $c^{\alpha,\beta}=1+12\beta^2$. 
When $c^{\alpha,\beta}>1$ and $h^{\alpha,\beta}:=\frac{1}{2}(\alpha^2+\beta^2)>0$, then 
one can redefine the scalar product on $V$ such that $\tilde{L}^{\alpha,\beta}$ becomes
a unitary highest weight representation with lowest energy $h^{\alpha,\beta}$ and normalized
highest weight vector $\Omega$. 
If $\alpha,\beta \in \RR$, then the above holds with no need of redefinition of the scalar product.
\end{proposition}
\begin{proof}
Straightforward check shows that $\tilde{L}^{\alpha,\beta}_n$ $(n\in\ZZ)$ indeed form a representation of the Virasoro algebra with central charge $c^{\alpha,\beta}$. The unitarity in case of real $\alpha,\beta$ parameters is also clear. Now let $M$ be the smallest invariant subspace for this representation that contain the vector $\Omega$. Since $\tilde{L}^{\alpha,\beta}_0 \Omega = h^{\alpha,\beta}\Omega$ and $\tilde{L}^{\alpha,\beta}_n \Omega=0$ for all $n>0$, we have that $\Omega$ is a highest weight vector and the restriction of the representation
to $M$ must factor through the Verma module corresponding to 
central charge $c^{\alpha,\beta}$ and lowest energy $h^{\alpha,\beta}$. However, as is  known \cite[Proposition 8.2]{kac}, for $c^{\alpha,\beta}>1$ and $h^{\alpha,\beta}>0$ the Verma module is irreducible, so actually the restriction of our representation to $M$ is equivalent to the Verma one. Thus the only thing that remains to be shown is that $M$ is the full space $V$.

Since $\tilde{L}^{\alpha,\beta}_0 = L_0 + h^{\alpha,\beta} I$, our subspace must be a direct sum $M=\oplus_{k=0}^\infty M_k$ where $M_k\subset V_k$ for each $k$. 
By what we have established ${\rm dim}(M_k)$ must be equal to the dimension of the $k^{\rm th}$ energy level of the Verma modul (corresponding to the value $k+h^{\alpha,\beta}$), which is the number of partitions of $k$. However, this is also the dimension of $V_k$; hence the inclusion $M_k\subset V_k$ is actually an equality and the proof is finished.  
\end{proof}
By what has been explained in the preliminaries, when $\alpha, \beta\in \RR$, the representation $\tilde{L}^{\alpha,\beta}$ gives rise to a stress-energy field $\tilde{T}_{\alpha,\beta}$ for which one has that 
\begin{equation}
\label{Tab}
\tilde{T}_{\alpha,\beta}(f)v = T(f)v + \alpha J(f)v + \beta J(f')v + \frac{\alpha^2+\beta^2}{2}\frac{1}{2\pi}\int_{-\pi}^{\pi}f\; v
\end{equation}
for every smooth vector $v$ and $f\in C^\infty(S^1,\RR)$.

\section{Dependence of expectations on lowest weight}

In what follows we shall fix a smooth function $g:S^1\to \RR$  
with the property that $\partial_\theta g|_{I} = 1$ for a certain open proper interval $I\Subset S^1$. (Note that on the full circle it is not possible to require the derivative to be constant $1$.) Then by a straightforward computation using the transformation rules (\ref{trule1}) and (\ref{trule2}), the formula (\ref{Tab}) and the fact that Weyl-operators preserve the 
set of smooth vectors and our fields in question are all essentially self-adjoint on an even smaller set, we find that for any $\alpha\in \RR$  and smooth function $f:S^1\to\RR$ with support in $I$ 
\begin{equation}
e^{i\alpha J(g)} \tilde{T}_{0,\beta}(f)
e^{-i\alpha J(g)} = \tilde{T}_{\alpha,\beta}(f).
\end{equation}
(Here we really mean {\it equality with domains}, not just an equality on a dense set).
In what follows, we shall also fix a collection $f_1,\ldots f_n$ of smooth, real-valued functions on $S^1$ with supports in $I$. For a $t\in \RR$ and an admissible pair of central charge $c$ and lowest weight $h$, we set
\begin{equation}
F(t,c,h) = \langle\Psi_{c,h}, \, 
e^{itT_{c,h}(f_1)} \ldots 
e^{itT_{c,h}(f_n)}\Psi_{c,h}
\rangle
\end{equation}
where $\Psi_{c,h}\in V_{c,h}$ is the
(up-to-phase) unique normalized highest weight vector. Note that $F(0,c,h)=1$ and that the
map $t\mapsto F(t,c,h)$ is obviously continuous.
\begin{lemma}
\label{analyticity1}
Suppose $c>1$. Then the function
$$
(\frac{c-1}{24},\infty) \ni h \mapsto F(t,c,h)
$$
extends in an analytic manner to the set $\{z\in\CC|\, {\rm Re}(z)>\frac{c-1}{24}\}$.
\end{lemma}
\begin{proof}
With $\beta = \sqrt{\frac{c-1}{12}}$ and $\alpha = \sqrt{2h-\beta^2}= \sqrt{2(h-\frac{c-1}{24})}$ we have that
\begin{eqnarray}
\label{real_alpha}
\nonumber
F(t,c,h) &=& \langle\Psi_{c,h}, \, 
e^{itT_{c,h}(f_1)} \ldots 
e^{itT_{c,h}(f_n)}\Psi_{c,h}\rangle \\
\nonumber
&=&\langle\Omega, \, 
e^{it\tilde{T}_{\alpha,\beta}(f_1)} \ldots 
e^{it\tilde{T}_{\alpha,\beta}(f_n)}\Omega\rangle\\
&=&\langle  e^{\overline{i\alpha} J(g)}\Omega, \, 
e^{it\tilde{T}_{0,\beta}(f_1)} \ldots 
e^{it\tilde{T}_{0,\beta}(f_n)}e^{-i\alpha J(g)}\Omega\rangle.
\end{eqnarray}
The claim then follows because $\Omega$ is an analytic vector
for $J(g)$.
\end{proof}
\begin{lemma}
$F(t,c,h)\, F(t,\tilde{c},\tilde{h}) = F(t,c+\tilde{c},h+\tilde{h})$.
\end{lemma}
\begin{proof}
The claim can be justified by considering tensorial products.
The restriction of the representation $U_{c,h}\otimes U_{\tilde{c},\tilde{h}}$ to the minimal invariant subspace containing the vector $\Psi_{c,h}\otimes \Psi_{\tilde{c},\tilde{h}}$ is unitarily equivalent to $U_{c+\tilde{c},h+\tilde{h}}$ as $\Psi_{c,h}\otimes \Psi_{\tilde{c},\tilde{h}}$ is a normalized highest weight vector for 
$U_{c,h}\otimes U_{\tilde{c},\tilde{h}}$ with energy $h+\tilde{h}$. So we may write
\begin{eqnarray}\nonumber
\langle\Psi_{c+\tilde{c},h+\tilde{h}},U_{c+\tilde{c},h+\tilde{h}}(\gamma)\Psi_{c+\tilde{c},h+\tilde{h}} \rangle &=&
\langle \Psi_{c,h}\otimes\Psi_{\tilde{c},\tilde{h}},
\left(U_{c,h}(\gamma)\otimes U_{\tilde{c},\tilde{h}}(\gamma)\right)
\Psi_{c,h}\otimes\Psi_{\tilde{c},\tilde{h}}\rangle 
\\
& =&
\langle \Psi_{c,h},U_{c,h}(\gamma)\Psi_{c,h}\rangle \,
\langle\Psi_{\tilde{c},\tilde{h}},U_{\tilde{c},\tilde{h}}(\gamma)
\Psi_{\tilde{c},\tilde{h}}\rangle
\end{eqnarray}
which however has the disadvantage, that --- since we deal with a {\it projective}, rather than a true representation --- the quantities appearing in it are only defined ``up-to-phase'' (i.e.\! a unit complex multiple). Nevertheless, using products of
exponentials of the form $e^{iT_{c,h}(f)}$ (rather than projective unitary operators of the form $U_{c,h}(\gamma)$) we can obtain similar equalities without the ambiguity of phases; and this is exactly what we wanted to justify. 
\end{proof}
\begin{corollary}
Let $c$ be greater than $1$. Then there exists an $\epsilon>0$
such that for every $t\in (-\epsilon,\epsilon)$, the function
$$
\RR^+ \ni h \mapsto F(t,c,h)
$$ 
has an analytical extension to the half plane $\{z\in\CC| {\rm Re}(z)>0\}$.
\end{corollary}
\begin{proof}
Let us choose a $c_0>1$ and an $h_0>\frac{c+c_0-1}{24}$. Since
$t\mapsto F(t,c_0,h_0)$ is continuous and $F(0,c_0,h_0)=1$, there
exists an $\epsilon>0$ such that $F(t,c_0,h_0)\neq 0$ for any 
$t\in (-\epsilon,\epsilon)$. Then, for such $t$ values, using our previous lemma we have that
\begin{equation}
F(t,c,h) = \frac{F(t,c+c_0,h+h_0)}{F(t,c_0,h_0)}.
\end{equation}
However, by lemma \ref{analyticity1}, the right hand side --- and hence also the left hand side --- has an analytical extension to
the region ${\rm Re}(h+h_0)>\frac{c+c_0-1}{24}$. In particular, 
all values of $h$ for which ${\rm Re}(h)>0$ are inside of this region. 
\end{proof}
\begin{proposition}
Let $c$ be greater than $1$. Then there exists an $\epsilon>0$
such that for every $t\in (-\epsilon,\epsilon)$ and $h\in (0,\frac{c-1}{24})$,
we have
$$
F(t,c,h) = \langle \eta_{-s}, 
e^{it\tilde{T}_{0,\beta}(f_1)} \ldots 
e^{it\tilde{T}_{0,\beta}(f_n)}
\eta_s \rangle
$$
where $\eta_r = e^{r J(g)}\Omega$ and $s=\sqrt{2(\frac{c-1}{24}-h)}$.
\end{proposition}
\begin{proof}
The right hand side of the claim is analytic in $s$, and by 
eq.\! (\ref{real_alpha}),
for every $s=i \alpha$, $\alpha \in \RR$ it is equal to 
$F(t,c,\frac{c-1}{24}+\frac{1}{2}\alpha^2)$. Thus the claim follows from the uniqueness of analytical continuations and our previous corollary.
\end{proof}
\begin{corollary}
\label{z1z2vectors}
Let $c>1, \; h\in (0,h_c)$ where $h_c=\frac{c-1}{24}$ and fix an $I\Subset S^1$. Then 
there exist two vectors $\zeta_L,\zeta_R\in \h_{c,h_c} = \overline{V_{c,h_c}}$
such that
$$
\langle\Psi_{c,h}, \, 
e^{iT_{c,h}(f_1)} \ldots 
e^{iT_{c,h}(f_n)}\Psi_{c,h}\rangle
=\langle \zeta_L,e^{iT_{c,h_c}(f_1)} \ldots 
e^{iT_{c,h_c}(f_n)} \zeta_R \rangle 
$$ 
for any collection $f_1,\ldots f_n \in C^\infty(S^1,\RR^+_0) \cup C^\infty(S^1,\RR^-_0)$ of functions with support in $I$. 
\end{corollary}
\begin{proof}
We shall first deal with the case when all functions involved are
nonnegative.
By the result of Fewster and Hollands \cite{FH} and the assumed nonnegativity, the self-adjoint operators $T_{c,h}(f_j)$ and
$\tilde{T}_{0,\beta}(f_j)$ appearing in the previous proposition are all bounded from below. Thus for both sides of that equation, there exists an extension (for the ``$t$'' variable) which is continuous and analytical in the upper complex half plane. It follows that the equality there deduced for $t\in(-\epsilon,\epsilon)$ actually holds for all $t\in \RR$. Then the claimed equality of our corollary follows by setting $t=1$ and considering that with $\beta = \sqrt{\frac{c-1}{12}}$, the stress-energy field $\tilde{T}_{0,\beta}$ given on the Hilbert space of the U(1) current algebra is a unitary equivalent realization of $T_{c,h_c}$.

Let us now shortly discuss the more general case when some of the functions involved are nonpositive, whereas possibly some others are nonnegative.
Say for simplicity that $n=2$, and $f_1\geq 0$ and $f_2\leq 0$. Then by what has been already established, the functions $G_1$ and $G_2$ defined by
the formulas
\begin{equation}
G_1(t) = \langle\Psi_{c,h}, \, 
e^{iT_{c,h}(f_1)} e^{itT_{c,h}(f_2)}\Psi_{c,h}\rangle,
\end{equation}
and
\begin{equation}
G_2(t) = \langle \zeta_L,e^{iT_{c,h_c}(f_1)} 
e^{i t T_{c,h_c}(f_2)} \zeta_R \rangle
\end{equation}
coincide for all $t\leq 0$, since in that case 
$t T_{c,h}(f_2)$ and $t T_{c,h_c}(f_2)$ are bounded from below.
Moreover, because of the spectrum condition, both $G_1$
and $G_2$ have continuous extensions that are analytic --- this time on the
{\it lower} complex half plane. As before, it follows that $G_1(1)=G_2(1)$ and hence our statement remains true even when some of the functions involved are nonpositive rather than nonnegative. 
\end{proof}
\section{Proof of local equivalence}

When $c,h$ is an admissible pair, we shall set 
\begin{equation}
\label{defa^0}
\a_{c,h}^0(I) = \rm{Alg}^*\{e^{T_{c,h}(f)} | \, f\in C^\infty(S^1,\RR^+_0),\, {\rm Supp}(f)\subset I \}
\end{equation}
for every open proper interval $I\Subset S^1$.
Here by ``Alg$^*$'' we mean the ``star algebra generated''; i.e.\! the linear span (without taking closures) of products of finite many of the generating unitaries and their inverses. When in need of {\it von Neumann} algebras,
we shall consider the double commutant
\begin{equation}
\a_{c,h}(I) = \left(\a_{c,h}^0(I)\right)''.
\end{equation}
Evidently, the above defined algebras satisfy {\it isotony}; that is, 
$\a_{c,h}(I_1)\subset \a_{c,h}(I_2)$ whenever $I_1\subset I_2$.
Slightly less evidently, but we also have the important {\it covariance} relation $U_{c,h}(\gamma)\a_{c,h}(I)U^*_{c,h}(\gamma) = \a_{c,h}(\gamma(I))$. This is because of the transformation formula mentioned in the preliminaries: 
$U_{c,h}(\gamma)T_{c,h}(f)U_{c,h}(\gamma)^* = T_{c,h}((\tilde{\gamma}' f)\circ \gamma^{-1}) \, + $ a constant times the identity. Since $\tilde{\gamma}'$ is a strictly positive function, $(\tilde{\gamma}' f)\circ \gamma^{-1}$ remains a nonnegative function. 
Finally, we also have {\it irreducibility}:
\begin{equation}
\label{irred}
\mathop\cap_{I\Subset S^1}\a_{c,h}(I)' = \CC I.
\end{equation}
This is because of the mentioned simplicity of $\diff$. Indeed, exponentials of vector fields corresponding to functions in $C^\infty(S^1,\RR^+_0)\cup C^\infty(S^1,\RR^-_0)$ which are ``localized'' (i.e.\! whose support is contained in some open proper interval) evidently form a normal subgroup containing nontrivial elements. Hence this subgroup is actually the full group; thus for any $\gamma\in\diff$ there exists an $n\in\NN$, some open proper intervals $I_1,\ldots I_n\Subset S^1$ and 
$f_1,\ldots f_n \in C^\infty(S^1,\RR^+_0)\cup C^\infty(S^1,\RR^-_0)$ such that the support of $f_j$
is contained in $I_j$ 
--- implying that $e^{iT_{c,h}(f_j)}\in \a_{c,h}^0(I_j)$ ---
and  
\begin{equation}
U_{c,h}(\gamma) = e^{iT_{c,h}(f_1)}\ldots e^{iT_{c,h}(f_n)}.
\end{equation}
Thus the explained irreducibility property is a direct consequence of the irreducibility of the representation $U_{c,h}$.
\begin{corollary}
Let $I\Subset S^1$ be an open proper interval. Then the lowest energy vector $\Psi_{c,h}$ is cyclic and separating for $\a_{c,h}(I)$ (and hence also for the dense subalgebra $\a_{c,h}^0(I)$).
\end{corollary}
\begin{proof}
This is essentially the Reeh-Schlieder theorem; the proof
can be done almost exactly as for example in \cite{FrG}. We have every ingredient like {\it isotony, covariance} and {\it irreducibility}. Although in the cited paper the authors seemingly also make use of the invariance of the vacuum vector, a closer inspection reveals that the argument works with any vector as long as it is analytical in some strip along the real line for the self-adjoint generators of certain one-parameter groups $t\mapsto U_{c,h}(\gamma_t)$. By \cite[Theorem 3.3]{BDL} this condition is always satisfied for any vector which is analytical
for $L^{c,h}_0$; in particular it holds for any eigenvector of
$L^{c,h}_0$.
\end{proof}
For $h=0$ the algebras $\a_{c,0}(I)$ $(I\Subset S^1)$ form
a local conformal net on $S^1$ and hence they are all type
I\!I\!I$_1$ factors, see \cite{FrG}. Note that one usually introduces local algebras by setting
\begin{equation}
\a_c(I) = \{U_{c,0}(\gamma)|\,\gamma\in G_I\}''
\end{equation} 
which evidently contains the algebra we use: $\a_c(I)\supset \a_{c,0}(I)$. However, our choice also forms a conformal net, so using {\it Haag-duality} \cite{FrG}, one has that
\begin{equation}
\label{haag}
\a_c(I)' = \a_c(S^1\setminus I) \supset
\a_{c,0}(S^1\setminus I) = \a_{c,0}(I)' 
\end{equation} 
implying that we have containment also in the other direction and so in return that 
\begin{equation}
\a_c(I)=\a_{c,0}(I).
\end{equation}
 Another important thing to note is that the introduced algebras are type I\!I\!I$_1$ factors also in case $c\geq 1$ and $h=h_c=\frac{c-1}{24}$. This is because by \cite{BS-M}, in this case we already know to have a unitary operator $W_I$ such that 
$W_I  T_{c,0}(f)W_I^* = T_{c,h_c}(f)$
for all $f\in C^\infty(S^1,\RR)$ with support in $I$, implying that $W_I \a_{c,0}(I)W_I^* = \a_{c,h_c}(I)$.

For simplicity, from now --- with the exception of the last theorem --- we shall fix a single $c>1$ and an $h\in(0,h_c)$ for once and all, so that we will not need to repeat this act at every single statement. We have the following simple operator algebraic fact.
\begin{lemma}
Since $\a_{c,h_c}(I)$ is a type I\!I\!I factor in standard form,
any normal state on $\a_{c,h_c}(I)$ can be represented
by a vector which is also cyclic for $\a_{c,h_c}(I)$.   
\end{lemma}
\begin{proof}
Modular theory --- see e.g.\! \cite[Sect.\! 10]{alapkonyv} --- ensures the existence of {\it some} representing vector $\zeta$.
Now let $P'$ be the ortho-projection onto $\overline{\a_{c,h_c}(I)\zeta}$. Evidently, we have that 
$P'\in \a_{c,h_c}(I)'$, which --- still by modular theory --- 
is also a type I\!I\!I factor. Hence there exists a partial isometry $V'\in \a_{c,h_c}(I)'$ such that $V'^*V' =P'$ while
$V'V'^* = I$. It is an exercise to check that $V'\zeta$ is a
cyclic vector for $\a_{c,h_c}(I)$ giving the same state as $\zeta$.
\end{proof}
\begin{proposition}\label{zvector}
Let $I$ be an open proper interval of $S^1$ and $\zeta_L,\zeta_R$ the two vectors in $\h_{c,h_c}$
given by corollary \ref{z1z2vectors}. Then there exists
a single vector $z$ which is cyclic for $\a_{c,h_c}(I)$ 
and satisfies
$$
\langle \zeta_L, A \zeta_R\rangle  = \langle \zeta, A \zeta\rangle
$$
for all $A\in \a_{c,h_c}(I)$.
\end{proposition}
\begin{proof}
 By the equation given in the cited corollary where the two vectors $\zeta_L,\zeta_R$ were introduced, $\langle \zeta_L,\cdot\,  \zeta_R\rangle$ is actually a state (i.e.\! a positive, normalized functional) on the dense subalgebra $\a_{c,h_c}^0(I)$.
 Using Kaplansky's density theorem, which ensures that every positive element of $\a_{c,h_c}(I)$ is the strong limit of a sequence of positives in $\a_{c,h_c}^0(I)$ --- see e.g.\! \cite[Theorem 3.10]{alapkonyv} --- we conclude that $\langle \zeta_L,\cdot \,\zeta_R\rangle$ is also a state on $\a_{c,h_c}(I)$. By the form it is given, it is evidently a normal state. Thus our
 claim follows directly from the previous lemma.  
\end{proof}
Collecting what we have established and using usual 
constructions we arrive to the following conclusion.
\begin{corollary}
The formula
$$
e^{iT_{c,h}(f_1)}\ldots e^{iT_{c,h}(f_n)}\Psi_{c,h} \mapsto 
e^{iT_{c,h_c}(f_1)}\ldots e^{iT_{c,h_c}(f_n)}\zeta
$$
where $f_1,\ldots f_n \in C^\infty(S^1,\RR^+_0)\cup
C^\infty(S^1,\RR^-_0)$ have all their supports in a certain 
$I\Subset S^1$ and $\zeta$ is the vector appearing in corollary 
\ref{zvector}, defines a unitary operator $K_{h,I}$ such that
$$
K_{h,I} e^{i T_{c,h}(f)}K_{h,I}^* = e^{i T_{c,h_c}(f)}
$$
for all $f\in C^\infty(S^1,\RR^+_0)\cup C^\infty(S^1,\RR^-_0)$ with support in $I$.
\end{corollary}
\begin{lemma}
The unitary operator $K_{h,I}$ appearing in the last corollary satisfies the relation
$$
K_{h,I} T_{c,h}(f) K_{h,I}^* = T_{c,h_c}(f)
$$
not only for functions $f\in C^\infty(S^1,\RR^+_0)\cup C^\infty(S^1,\RR^-_0)$ with support in $I$, but actually for
any $f\in C^\infty(S^1,\RR)$ with support in $I$.
\end{lemma}
\begin{proof}
For any $f\in C^\infty(S^1,\RR)$ with support in $I$ we can find two nonnegative smooth functions $f_1,f_2\geq 0$ with support still in $I$ such that $f=f_1-f_2$. Of course the relations $K_{h,I}T_{c,h}(f_j) K_{h,I}^* = T_{c,h_c}(f_j)$
$(j=1,2)$ evidently follow from our last corollary. Then using eq.\! (\ref{alg_rel}),
\begin{eqnarray}
\nonumber
K_{h,I}T_{c,h}(f)K_{h,I}^* &=& 
K_{h,I}\left(\overline{T_{c,h}(f_1)-T_{c,h}(f_2)}\right)K_{h,I}^*\\
&=&
\nonumber
\overline{K_{h,I} T_{c,h}(f_1) K_{h,I}^* -
 K_{h,I}T_{c,h}(f_2) K_{h,I}^*}\\
 &=& 
 \overline{T_{c,h_c}(f_1) - T_{c,h_c}(f_2)} = T_{c,h_c}(f)
\end{eqnarray}
which is what we wanted to prove.
\end{proof}
Collecting all we have obtained so far, we can now state the main result of this paper. 
\begin{theorem}
Let $c,h$ and $\tilde{c},\tilde{h}$ be two admissible pairs of central charges and highest weights. Then $U_{c,h}$ is locally equivalent to $U_{c,\tilde{h}}$ if and only if $c=\tilde{c}$.
\end{theorem}
\begin{proof}
As was already noted at lemma \ref{lemmacc}, $c=\tilde{c}$ is a necessary condition of local equivalence. Moreover, for each admissible pair of the central charge $c$ and highest weight $h$, the pair $(c,0)$ is also admissible. Clearly then, it is
enough to prove the local equivalence between $U_{c,h}$ and
$U_{c,0}$; if we can do that for any of the possible $h$-values,
than --- passing through $U_{c,0}$ --- we can also conclude the local equivalence of $U_{c,h}$ and $U_{c,\tilde{h}}$.
 
As was explained in the introduction, apart from the region
$\{1<c,\, 0<h<h_c=\frac{c-1}{24}\}$ we already know that 
for any $I\Subset S^1$ there exists a unitary operator $W_{h,I}$
such that 
\begin{equation}
\label{WTW}
W_{h,I} T_{c,0}(f) W_{h,I}^* = T_{c,h}(f)
\end{equation}
for all $f\in C^\infty(S^1,\RR)$ with support in $I$.
On the other hand, even if $c,h$ is in this ``bad''
region, we can use 1) the unitary $K_{h,I}$ 
constructed in this section and 2) the fact that
$c,h_c$ lies outside of the ``bad'' region so we already 
have a unitary $W_{h_c,I}$. Then
\begin{equation}
W_{h_c,I} T_{c,0}(f) W_{h_c,I}^* = T_{c,h_c}(f)
= K_{h,I} T_{c,h}(f) K_{h,I}^*
\end{equation}
implying that with $W_{h,I}:=K_{h,I}^* W_{h_c,I}$ 
we satisfy eq.\! (\ref{WTW}). Thus, the existence of 
a ``suitable'' unitary $W_{h,I}$ is ensured in all cases.
\end{proof}
\bigskip

\section*{Acknowledgment}

The author would like to thank the numerous discussions on the topic with Sebastiano Carpi, Roberto Longo and Yoh Tanimoto.


\begin{thebibliography}{99}
 
 \bibitem{BDL} D.\! Buccholz, C.\! D'Antoni and R.\! Longo:
 Nuclearity and Thermal States in Conformal Field Theory. 
  {\it Commun.\! Math.\! Phys.} {\bf 270} 
  (2007), 267--293.
  

 \bibitem{BS-M} D.\! Buchholz and H.\! Schulz-Mirbach: 
 Haag duality in conformal quantum field theory. {\it Rev.\! Math.\! Phys.} {\bf 2} (1990), 105--125.
 
 \bibitem{CLTW} P.\! Camassa, R.\! Longo, Y.\! Tanimoto and
  M.\! Weiner: Thermal States in Conformal QFT. II. {\it Commun.\! Math.\! Phys.} {\bf 315} (2012), 771--802. 
 
 
 \bibitem{Carpi2} S.\! Carpi: The Virasoro algebra and sectors with infinite statistical dimension. {\it Ann.\! Henri Poincar\'e} {\bf 4} (2003), 601--611.
 

 \bibitem{Carpi1} S.\! Carpi: On the representation theory of Virasoro  nets. {\it Commun.\! Math.\! Phys.} {\bf 244} (2004), 261--284.  
 
 \bibitem{FH} C.\! J.\! Fewster and S.\! Hollands: 
 Quantum energy inequalities in two-dimensional conformal field theory. {\it Rev.\! Math.\! Phys.} 
 {\bf 17} (2005), 577--612.
 
\bibitem{todorov} P.\! Furlan, G.\! M.\! Sotkov and I.\! T.\! Todorov: Two-dimensional conformal quantum field theory.
 {\it Riv. Nuovo Cimento} {\bf 12} (1989), 1--202.
 
 \bibitem{FrG} F.\! Gabbiani and J.\! Fr\"{o}hlich: Operator algebras and
 conformal field theory. {\it Commun.\! Math.\! Phys.} {\bf 155} (1993),  569--640.
 
 \bibitem{GKO} P.\! Goddard, A.\! Kent and D.\! Olive: Unitary representations of 
 the Virasoro and super-Virasoro algebra. {\it Commun.\! Math.\! Phys.} {\bf 103} (1986), 105--119.
 

 \bibitem{GoWa} R.\! Goodman and N.\! R.\! Wallach: Projective unitary
 positive-energy representations of ${\rm Diff}(S^1)$. 
 {\it J.\! Funct.\! Anal.} {\bf 63}, 299--321 (1985).
 
\bibitem{kac} V.\! G.\! Kac and A.\! K.\! Raina:
 Bombay Lectures on Highest Weight Representations of Infinite Dimensional Lie Algebras.{ \it World Scientific}, Singapore, 1987.
 
 \bibitem{KL}Y.\! Kawahigashi and R.\! Longo: Classification of local conformal
 nets. Case $c < 1$. {\it Ann.\! of Math.} {\bf 160} (2004), 493--522.
 
 \bibitem{Loke} T.\! Loke: Operator algebras and conformal field theory
 of the discrete series representation of $\diff$. PhD Thesis, University 
 of Cambridge, 1994.
 
 
 \bibitem{Milnor} J.\! Milnor: Remarks on infinite-dimensional Lie groups. 
  In B.S. De Witt and R. Stora Eds.: {\it Relativity, groups and topology 
  II.} Les Houches, Session  XL, 1983, Elsevier, Amsterdam, New York, 1984, 
  pp. 1007--1057.
  
 \bibitem{Rehren} K.\! H.\! Rehren: A new view of the Virasoro algebra.
  {\it Lett.\! Math.\! Phys.} {\bf 30} (1994), 125--130. 
 
 
 \bibitem{alapkonyv} \unichar{536}.\! Str\unichar{259}til\unichar{259}, L.\! Zsid\'o: Lectures
 on von Neumann algebras. Editura Academiei and Abacus Press, Kent 1979.
 
 \bibitem{weiner}M.\! Weiner: Conformal covariance and positivity of energy in charged sectors. {\it Commun.\! Math.\! Phys.} {\bf 265} (2006), 493--506.
 
 \bibitem{Xu}F.\! Xu: Strong additivity and conformal nets. {\it Pacific J.\! Math.} {\bf 221} (2005), 167--199.
 
 
 \end{thebibliography}
\end{document}